\documentclass[a4paper,12pt]{article}

\usepackage{amsmath,amssymb,mathrsfs}
\usepackage[dvips]{graphicx}
\usepackage{amsthm,bm,
color}

\makeatletter
\def\figcaption{\def\@captype{figure}\caption}
\makeatother

\makeatletter
 
  \@addtoreset{equation}{section}
 \makeatother

\newtheorem{theorem}{\bf Theorem}[section]
\newtheorem{lemma}[theorem]{\bf Lemma}
\newtheorem{proposition}[theorem]{\bf Proposition}
\newtheorem{corollary}[theorem]{\bf Corollary}

\newtheorem{example}{\bf Example}[section]
\newtheorem{remark}{\bf Remark}[section]
\newtheorem{definition}{\bf Definition}[section]

\title{Spectral mapping theorem \\ of an abstract quantum walk}
\author{
Yusuke Higuchi
\thanks{Mathematics Laboratories, College of Arts and Sciences, Showa University,
4562 Kamiyoshida, Fujiyoshida, Yamanashi 403-0005, Japan,
email: higuchi@cas.showa-u.ac.jp
},
Etsuo Segawa
\thanks{Graduate School of Information Sciences, Tohoku University,
Aoba, Sendai 980-8579, Japan,
email: e-segawa@m.tohoku.ac.jp
	}, and
Akito Suzuki
\thanks{Division of Mathematics and Physics, 
Faculty of Engineering, Shinshu University, Wakasato, Nagano 380-8553, Japan,
e-mail: akito@shinshu-u.ac.jp
	}
}
\begin{document}

\maketitle

\begin{abstract}
Given two Hilbert spaces, $\mathcal{H}$ and $\mathcal{K}$,
we introduce an abstract unitary operator $U$ on $\mathcal{H}$
and its discriminant $T$ on $\mathcal{K}$ induced by 
a coisometry from $\mathcal{H}$ to $\mathcal{K}$
and a unitary involution on $\mathcal{H}$.
In a particular case,
these operators $U$ and $T$ become the evolution operator of the Szegedy walk 
on a graph, possibly infinite, and the transition probability operator thereon.
We show  
the spectral mapping theorem between $U$ and $T$ 
via the Joukowsky transform. 
Using this result, we have 
completely detemined the spectrum of the Grover walk on the 
Sierpi\'nski lattice, 
which is pure point and has a Cantor-like structure.
\end{abstract}
\section{Introduction}
Quantum walks, whose primitive form appeared in \cite{FH} (1965) and \cite{Gu} (1988),
attracted the attention of many researchers 
at the beginning of the century because of their efficiencies of 
the quantum speed-up of search algorithm on some graphs 
(see \cite{Am0} and its references.).  
Szegedy  \cite{Sz} (2004) introduced an inclusive class of quantum walks 
partially including previous quantum walk models \cite{SKW,Am1,AKR}.
One of the interesting aspects of this class is that the spectrum of a walk 
is reduced to a spectral analysis in terms of the underlying reversible random walks 
on the same graph. 
This spectral mapping theorem is sometimes quite useful not only 
in estimating the efficiency of a search algorithm based on quantum walks \cite{SKW,Sz} 
but also in characterizing its stochastic long-time behavior \cite{IKS, KOS}. 

Recently, an extended version of the walk, 
the twisted Szegedy walk, 
was introduced in \cite{HKSS14}. 
For a graph $G=(V,D)$ with vertices $V$ and symmetric arcs $D$, 
the time evolution $U^{(w,\theta)}$ of the twisted Szegedy walk on $G$
is a unitary operator on $\ell^2(D)$ defined by 
\[ U^{(w,\theta)} = S^{(\theta)}C^{(w)}
\quad \mbox{with $C^{(w)} = 2 d_A^* d_A - 1$}. \]
Here $S^{(\theta)}$ is 
called a shift operator
and is a unitary involution
defined from a 1-form $\theta:D \to \mathbb{C}$.
$C^{(w)}$ is a coin operator and
$d_A:\ell^2(D) \to \ell(V)$ is a boundary operator, 
which is 
a coisometry defined from a weight $w:D \to \mathbb{C}$.
For a particular choice of $\theta$ and $w$,
$U^{(w,\theta)}$ becomes the evolution $U_G$ of the Grover walk on $G$,
which is one of the most intensively studied model of quantum walks on graphs (see \cite{W, Am0, HKSS13} and the references therein).
The discriminant $T^{(\theta,w)}= d_A S^{(\theta,w)} d_A^*$
is a self-adjoint operator on $\ell^{2}(V)$.
In the case of the Grover walk on $G$,
the discriminant of $U_G$ is unitary equivalent 
to the transition 
probability 
operator $P_G$
of the symmetric random walk on $G$, 
in which a walker on a vertex
moves to a neighbor vertex with isotropic probability.  
In \cite{HKSS14} the following spectral mapping 
theorem by the Joukowsky transform $\varphi(x)=(x+x^{-1})/2$
was obtained for finite graphs,
{\it i.e.}, $|V|, |D| < \infty$:
\begin{equation}
\label{fmap}
\sigma_{\rm p}(U^{(w,\theta)})
	= \varphi^{-1}(\sigma_{\rm p}(T^{(w,\theta)}))
		\cup \{1 \}^{M_+} \cup \{-1\}^{M_-}, 
\end{equation}
where 
$M_\pm = {\dim} \ker(d_A) \cap \ker(S^{(w, \theta)} \pm 1)$
and $\sigma_{\rm p}(\cdot)$ denotes the set of all eigenvalues.
In the expression above, $\{\pm 1\}^{M_{\pm}}$ implies the 
set of eigenvalue $\pm 1$ of multiplicity $M_{\pm}$, respectively; 
we assume 
$\{\pm 1\}^{M_\pm}=\emptyset$  
if $M_{\pm}=0$. 
Using \eqref{fmap}, the spectra of the evolution of the Grover walk
on crystal lattices, which have finite quotient graphs,
were also obtained.

In this paper, 
we extend the above spectral mapping \eqref{fmap} for 
finite graphs
to that for general infinite graphs. 
To this end, once we discard the graph structure,
consider two arbitrary Hilbert spaces $\mathcal{H}$ and $\mathcal{K}$,
and define an abstractive unitary operator $U$ on $\mathcal{H}$ 
as 
\begin{equation}
\label{formU1}
U=S(2d_A^*d_A-1). 
\end{equation} 
We suppose that:
(1) $S$ is a unitary involution on $\mathcal{H}$;
(2) $d_A$ is a coisometry from $\mathcal{H}$ to $\mathcal{K}$. 
Then, we obtain the spectral mapping theorem 
between $U$ and the discriminant $T=d_A S d_A^*$ of $U$,
which is a self-adjoint contraction operator 
on $\mathcal{K}$. 
Let $M_\pm = {\rm dim}\ker(d_A) \cap \ker (S\pm 1)$. 
We use $\sigma(\cdot)$ to denote the spectrum.

\begin{theorem}
\label{theorem1}
{\rm
Let $U$ and $T$ be as above. 
Then, 
	\begin{align}
       \sigma(U) 
        	&= \varphi^{-1}(\sigma(T)) \cup \{1\}^{M_+}
        		 \cup \{-1\}^{M_-}, \label{sigmap1}\\
        \sigma_{\rm p}(U) 
        	&= \varphi^{-1}(\sigma_{\rm p}(T)) \cup \{1\}^{M_+}
        		 \cup \{-1\}^{M_-}.  \label{sigmap2}
        \end{align}
Moreover, for all $\lambda \in \sigma_{\rm p}(U)$,
\begin{equation*}
{\rm dim}\ker(U-\lambda) 
	= \begin{cases}
		{\rm dim}\ker(T-\varphi(\lambda)), 
			& \lambda \not= \pm 1 \\
		 {\rm dim}\ker(T \mp 1) + M_\pm,
			& \lambda = \pm 1.
		\end{cases}
\end{equation*}
}
\end{theorem}

In a companion paper \cite{SS},
we construct 
the generator of $U$ under the conditions (1) and (2). 
As a byproduct \cite[Theorem 4.1 and Corollary 4.3]{SS}, 
the continuous part of $U$
is  unitarily equivalent to the continuous part of
${\rm exp}(i \arccos T) \oplus 
	{\rm exp}(-i \arccos T)$. %
Combining this with Theorem \ref{theorem1}
yields the following corollary.
We denote by $\sigma_{\rm c}(\cdot)$, 
$\sigma_{\rm ac}(\cdot)$, and $\sigma_{\rm sc}(\cdot)$
the continuous, absolutely continuous,
and  singular continuous spectrum.
\begin{corollary}
\label{Cor}
{\rm
Let $U$ and $T$ be as above. Then
\begin{equation*}
\sigma_\sharp(U) = \varphi^{-1}(\sigma_\sharp(T))
\quad \mbox{for $\sharp = {\rm c, ac, sc}$}.
\end{equation*}
}
\end{corollary}
As long as the conditions (1) and (2) are satisfied, 
Theorem \ref{theorem1} and Corollary \ref{Cor} ensure that the spectral mapping theorem holds not only for the Szegedy walks
on finite graphs but general infinite graphs, 
and also for arbitrary unitary operators of the form 
\eqref{formU1}. 
An example which is not directly concerned with a graph
is given in Section 3.

In the rest of this section, 
we go back to the graph world and give some interesting examples
of the Grover walks on infinite graphs $G$.
As mentioned, 
the discriminant of the Grover evolution $U_G$ is unitarily equivalent 
to the transition probability operator $P_G$.
See Example~\ref{ExampleTSW} for the details of the graph setting.
\\
\indent
First we see that Theorem \ref{theorem1} recovers 
some result in \cite{HKSS14} for a crystal lattice $G$
such as the $d$-dimensinal lattice $\mathbb{Z}^{d}$,
the hexagonal, the triangular, and the Kagome ones.
Detailed spectral structures, 
including the multilicities of eigenvalues $\pm 1$, 
are described in terms of geometric properties of a graph, which 
can be seen in \cite{HKSS14}. 
The continuous spectrum of $U_G$ is obtained by Corollary \ref{Cor}.
$P_G$ does not have any singular continuous spectrum on
the crystal lattice (\cite{GeNi, HiNo}),
then neither does $U_{G}$.
\begin{example}[\cite{HKSS14}]
\label{crystal}
Let $G$ be a crystal lattie with a finite quotient graph. Then 
\begin{align*}
& \sigma(U_G)= \varphi^{-1}(\sigma(P_G))
	\cup\{1\}^{M_+} \cup \{-1\}^{M_-}, \\
& \sigma_{\rm p}(U_G)= \varphi^{-1}(\sigma_{\rm p}(P_G))
	\cup\{1\}^{M_+} \cup \{-1\}^{M_-}, \\
& \sigma_{\rm ac}(U_G) = \varphi^{-1}(\sigma_{\rm ac}(P_G)),
	\quad \sigma_{\rm sc}(U_G) = \emptyset,
\end{align*}
where $M_\pm = \infty$ if $G$ has a cycle; $0$ otherwise.
\end{example}

Next two examples may be typical ones for showing 
the advantage of Theorem~\ref{theorem1}. 
The results in \cite{HKSS14} cannot be applied to them.

Let $\mathcal{T}_{d}$
be a $d$-regular tree $\mathcal{T}_{d}$ ($d\geq 2$), 
which is an infinite {\it acyclic\/} graph of constant degree $d$. 
See Figure~\ref{fig:figure1}. 
The spectrum of the transition probability
operator 
$P_{\mathcal{T}_{d}}$ on $\mathcal{T}_{d}$
is $\sigma(P_{\mathcal{T}_{d}})=[-2\sqrt{d-1}/d, 2\sqrt{d-1}/d]$. 
Refer to \cite{Figa,Sunada}, for instance. 
If $d\geq 3$, $\mathcal{T}_{d}$ is not a crystal lattice 
but the spectral mapping still holds from Theorem \ref{theorem1}. 
Detailed geometrical and analytical structure of the eigenvalues $\pm 1$ 
are discussed in \cite{HSe}. 
Moreover, we can find $\sigma(U_{\mathcal{T}_{d}})$ 
has no singular continuous spectrum 
by results in \cite{Figa,Sunada} with Corollary~\ref{Cor}. 
\begin{example}
\label{tree}
For the $d$-regular tree $\mathcal{T}_{d}$,
\begin{align*}
& \sigma(U_{\mathcal{T}_{d}})
= \sigma_{\rm ac}(U_{\mathcal{T}_{d}}) 
	\cup\sigma_{\rm p}(U_{\mathcal{T}_{d}}), 
\quad \sigma_{\rm p}(U_{\mathcal{T}_{d}})
	= \{1\}^{M_+} \cup \{-1\}^{M_-},\\
& \sigma_{\rm ac}(U_{\mathcal{T}_{d}})
	= \varphi^{-1}([-2\sqrt{d-1}/d, 2\sqrt{d-1}/d]), 
	\quad \sigma_{\rm sc}(U_{\mathcal{T}_{d}}) = \emptyset,
\end{align*}
where $M_\pm = \infty$ if $d\geq 3$; $0$ if $d=2$. 
\end{example}

\begin{figure}
\begin{minipage}{0.4\hsize}
\begin{center}
\includegraphics[scale=.3]{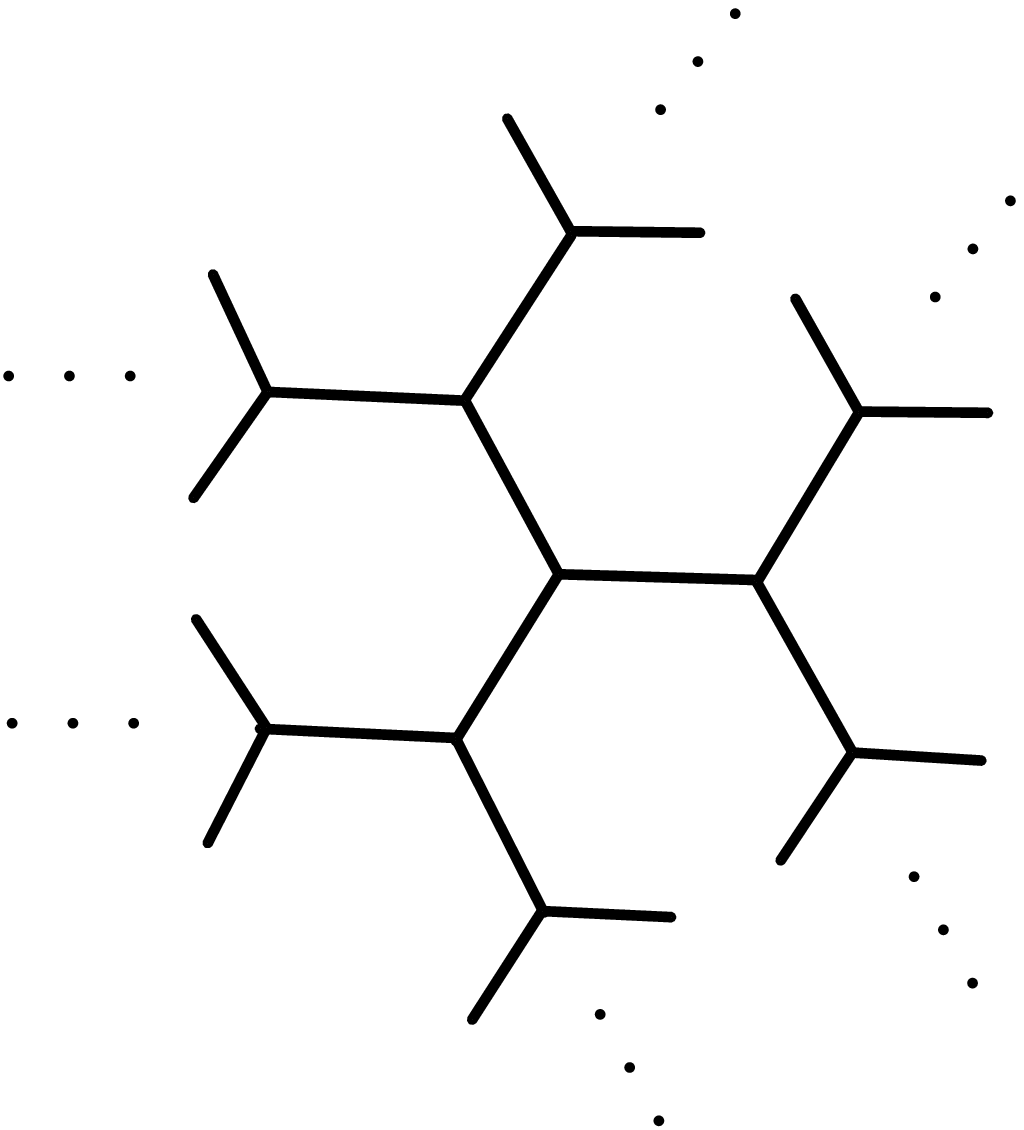}
\end{center}
\caption{$\mathcal{T}_{d}$ ($d=3$).}
\label{fig:figure1}
 \end{minipage}
\begin{minipage}{0.6\hsize}
\centering
\includegraphics[scale=0.68]{lattice.eps}
\caption{$\mathcal{S}_{d}$ ($d=2$).}
\label{fig:figure2}
\end{minipage}
\end{figure}
\par

Last example in this section is a graph which could be said to be
a skeleton of the famous fractal figure. Here we call it 
the $d$-dim Sierpi\'nski lattice $\mathcal{S}_{d}$, which can be 
found in \cite{HS,BP}. See Figure~\ref{fig:figure2}.

To construct an infinite Sierpi\'nski lattice $\mathcal{S}_{d}$, 
we define $f_{i} \ : \ \mathbb{R}^{d} \to \mathbb{R}^{d}$ 
by $f_{i}(x) = \frac{1}{2}(x+e_i) \ (0 \leq i \leq d)$. 
Here $\{e_i\}_{i=1}^{d}$ is the standard basis of $\mathbb{R}^d$ and 
$e_0$ be the $0$ vector. 
Furthermore we define $V_n$ inductively as follows:
\[
 V_0 = \bigcup_{0 \leq i < j \leq d} \{(1-t)e_i + t e_j \in \mathbb{R}^d \ ;  \ 
 0 \leq t \leq 1 \}
\]
and 
\[
 V_n = f_0^{-1}\left(\bigcup_{i=0}^d f_i(V_{n-1})\right),  
\quad n \geq 1.  
\]
We regard $\widetilde{\mathcal{S}}_d = \cup_{n \ge 0} V_n$ as an infinite
graph which is $2d$-regular except at the origin and the degree of the 
origin is $d$. 
Here the set of vertices of $V_{0}$ is identified with 
 $\{e_i\}_{i=0}^d$ and $V(\widetilde{\mathcal{S}}_d)$ 
with the set of all vertices 
defined repeatedly;  
similarly, the set of unoriented edges 
of
$V_{0}$ is identified with 
$\{e_{i}e_{j}\}_{0\leq i<j\leq d}$ and $E(\widetilde{\mathcal{S}}_d)$ 
with the set of all vertices 
defined repeatedly. 
We prepare two copies of an infinite graph $\widetilde{\mathcal{S}}_d$ 
and identify the vertices 
(the origins) of degree $d$. 
We call the infinite $2d$-regular graph constructed here 
the $d$-dimensional Sierpi\'nski lattice and denote it by $\mathcal{S}_d$. 
For such a fractal graph, we fortunately know the spectrum  
of the transition probability
operator $P_{\mathcal{S}_d}$ of the symmetric 
random walk on $\mathcal{S}_d$.
Refer to  \cite{FS,T,HS} for instance. 
Remark that $\mathcal{S}_d$ is not a crystal lattice. 
By Theorem \ref{theorem1}, 
we obtain the following.
\begin{example}\label{sierpinski}
For the $d$-dimensional Sierpi\'nski lattice  $\mathcal{S}_d$ with 
$d \geq 2$,
\begin{align*}
\sigma(U_{\mathcal{S}_d})
& =\varphi^{-1}\Big(
\overline{\bigcup_{k=0}^{\infty} \left\{ 
\{1-\rho^{-k}(\frac{d+1}{2d})\} \cup \{1-\rho^{-k}(\frac{d+3}{2d})\}
\right\}} \cup \{ \frac{-1}{d} \}
\Big) \\
& \qquad \cup  \{1\}^{M_+} \cup \{-1\}^{M_-}, 
\end{align*}
where $\rho(x) =-2dx^2+(d+3)x$
and $M_\pm = \infty$. 
\end{example}
\noindent
In the above, 
the multiplicity $M_\pm$ of $\{\pm 1\}$ can be derived 
by
the same argument as in \cite{HKSS14} 
in terms of the distribution of cycles. 
We remark that the same results hold for a standard Sierpi\'nski lattice
$\widetilde{\mathcal{S}}_{2}$. See \cite{T}.

We close this section
by mentioning a typical stochastic behavior named localization
for the above three examples of the Grover walk 
on an infinite graph $G=(V,D)$. 
Let $\psi_n = U_G^n\psi_0$ be the state of a walker 
at time $n \in \mathbb{N}$ with the initial state $\psi_0$ ($\|\psi_0\|=1$). 
The distribution $\mu_n^{(\psi_0)}: V\to [0,1]$ of the walker at time $n$
is defined by $\mu_n^{(\psi_0)}(u) =\sum_{e:t(e)=u}|\psi_n(e)|^2$. 
We say localization occurs if 
$\limsup_{n \to \infty}\mu_n^{(\psi_0)}(u) > 0$
with some $u \in V$.
It follows from the result of Teplyaev \cite{T}
that the spectrum of $P_{\mathcal{S}_2}$ is pure point 
and 
hence, by Theorem \ref{theorem1} and  Corollary \ref{Cor},
so is that of 
$U_{\mathcal{S}_2}$.
In particular, 
$U_{\mathcal{S}_{2}}$ has a complete set of eigenvectors.
By \cite[Corollary 4.4]{SS}, localization occurs 
for any initial state $\psi_0$.  
Thus, the time-evolution behavior of the Grover walk on $\mathcal{S}_{2}$ consists of only ``localization''. 
From Example \ref{crystal},
$\mathbb{Z}^d$ ($d \geq 2$)
satisfies
$\sigma_{\rm ac}(U_{\mathbb{Z}^d}) = \varphi^{-1}([-1, 1]) = S^1$,
$\sigma_{\rm sc}(U_{\mathbb{Z}^d})  = \emptyset$, 
and 
$\sigma_{\rm p}(U_{\mathbb{Z}^d}) = \{\pm 1\}$,  
because $\mathbb{Z}^d$ includes cycles. 
Example \ref{tree} also concludes that 
$S^1 \supsetneq \sigma_{\rm ac}(U_{\mathcal{T}_d}) \not= \emptyset$, 
$\sigma_{\rm sc}(U_{\mathcal{T}_d}) = \emptyset$, 
$\sigma_{\rm p}(U_{\mathcal{T}_d}) = \{\pm 1\}$ 
if $d \geq 3$. 
Hence, the time-evolution behavior of the Grover walk on
$\mathbb{Z}^d$ and $\mathcal{T}_d$ ($d \geq 2$) 
have a possibility to exhibit another stochastic behavior, for example, a linear spreading. 
Because $\mathcal{T}_2 = \mathbb{Z}$, it follows that
$\sigma_{\rm ac}(U_\mathbb{Z})
	= \varphi^{-1}([-1,1]) = S^1$ 
and $\sigma_{\rm sc}(U_\mathbb{Z})
	= \sigma_{\rm p}(U_\mathbb{Z}) = \emptyset$. 
Hence, localization never occurs for any initial states.
We summarize spectral and localization properties for the above three examples in the following table.

\begin{center}
\begin{tabular}{|c||c|c|c|c|}
\hline
Graph $G$
& $\sigma_{\rm ac}(U_G)$
& $\sigma_{\rm sc}(U_G)$
& $\sigma_{\rm p}(U_G)$  &  Localization \\ \hline \hline
$\mathbb{Z}$
& $S^1$
& $\emptyset$
& $\emptyset$
& for $\not\exists \psi_0 \in \ell^2(D)$ \\ \hline
$\mathbb{Z}^d$ ($d \geq 2$)
& $S^1$
&  $\emptyset$
& $\{\pm 1\}$
& for $\exists \psi_0 \in \ell^2(D)$ \\ \hline
$\mathcal{T}_d$ ($d \geq 3$)
& $\subsetneq S^1$
& $\emptyset$        
&   $\{\pm1 \}$
& for $\exists \psi_0 \in \ell^2(D)$ \\ \hline
$\mathcal{S}_{2}$
& $\emptyset$
& $\emptyset$
& $\sigma(U_G)= \overline{\sigma_{\rm p}(U_G)}$
& for $\forall \psi_0 \in \ell^2(D)$ \\
\hline
\end{tabular}
\end{center}

\begin{remark}
{\rm
In the above examples,
we consider the Grover walk \cite{Am0,W}.  
The behavior of quantum walks strongly depends 
on the definitions of shift and coin operators. 
Indeed, for 
other types of quantum walks 
on $\mathcal{S}_2$ and its Sierpi{\'n}ski pre-lattice, 
a numerical simulation suggests a diffusive spreading rate~\cite{LP}, 
and their recurrence relation 
obtained by a notion of renormalization group suggests that the spreading rate is close to ballistic~\cite{BFP}.
}
\end{remark}

This paper is organized as follows. 
We prepare notations and provide our setting in Section 2. 
Under the setting, we construct an abstractive unitary operator on $\mathcal{H}$ denoted by a unitary involution $S$ and coisometry map $d_A$ 
and give two examples in section 3. 
In Section 4, we introduce important invariant subspaces of our abstractive quantum walk induced by $S$ and $d_A$. 
We give the proofs of
Eqs.~(\ref{sigmap2}) and (\ref{sigmap1})
of Theorem~1 in Sections 5 and 6, respectively. 
The final section is the summary and discussion. 
\section{Preliminaries}
Let $\mathcal{H}$ and $\mathcal{K}$ be complex Hilbert spaces
and $d_A:\mathcal{H} \to \mathcal{K}$ a coisometry,
{\it i.e.}, $d_A$ is bounded and
\begin{equation}
\label{ddstar}
d_A d_A^* = I_{\mathcal{K}}, 
\end{equation}
where $I_{\mathcal{K}}$ is the identity operator on $\mathcal{K}$.
Then, $d_A^*:\mathcal{K} \to \mathcal{H}$ is an isometry,
because
\[ \|d_A^* f\|_{\mathcal{H}}^2 = \langle f, d_A d_A^* f \rangle_{\mathcal{K}} 
	= \|f\|_{\mathcal{K}}^2, \quad f \in \mathcal{K}. \]
Because $(d_A^*d_A)^2=d_A^* d_A$ from \eqref{ddstar},
we know that $\Pi_\mathcal {A}:= d_A^*d_A$
is the projection onto  $\mathcal{A}:={\rm Ran} (d_A^* d_A)$.
By \eqref{ddstar} again, we have
$f = d_A(d_A^*f) \in d_A\mathcal{H}$ for $f \in \mathcal{K}$.
Hence we observe that $\mathcal{K} = d_A \mathcal{H}$ and
\[ \mathcal{A} = d_A^* \mathcal{K} 
	= {\rm Ran} (d_A^*) = {\rm Ran} \Pi_\mathcal{A}. \]
Because $\ker(d_A) = {\rm Ran}(d_A^*)^\bot$,
we know that $d_A$ is a partial isometry and
\[ \|d_A \psi \|_{{\mathcal{K}}}^2
		= \langle \psi, \Pi_\mathcal{A} \psi \rangle_{{\mathcal{H}}}
		= \|\psi\|_{{\mathcal{H}}}^2,
			\quad \psi \in \mathcal{A} = \ker(d_A)^\perp. \]

We call the self-adjoint operator $C := 2 d_A^* d_A - 1$ on $\mathcal{K}$ 
a {\it coin operator}, 
because we observe, from Lemma \ref{lemmaC}, that
$C$ is decomposed into 
\[ C = I \oplus (-I) \quad \mbox{on 
$\mathcal{H} =\mathcal{A} \oplus \mathcal{A}^\perp$.} \]
\begin{lemma}
\label{lemmaC}
{\rm
Let $d_A$ and $C$ be as above.
Then, we have the following:\\
(1) $\sigma(C) = \{\pm 1\}$.\\
(2) $\mathcal{A} = \ker(C -1)$ and $\mathcal{A}^\perp = \ker(C+1)$\\
(3) $\displaystyle P_{\pm} =  \frac{1 \pm C}{2}$ 
	is the projection onto $\ker(C \mp 1)$ and
\[ P_+ = \Pi_\mathcal{A}, \quad P_- = \Pi_{\mathcal{A}^\perp}. \]
(4) 
$\mathcal{H}$ is decomposed into
\begin{align*}
\mathcal{H} & =  {\rm Ran} (d_A^*) \oplus \ker (d_A)
	 = \mathcal{A} \oplus \mathcal{A}^\perp 
	 = {\rm Ran} P_+ \oplus {\rm Ran} P_-.
\end{align*}
\quad In particular, we have
\[ C d_A^* = d_A^*, \quad  d_A C = d_A. \]
}
\end{lemma}
\begin{proof}
{\rm
(1) is proved by the self-adjointness of $C$ and the fact
\[ C^2 = (2 d_A^* d_A - 1)^2 = 1. \]
By direct calculation, we have $C \Pi_\mathcal{A} = \Pi_\mathcal{A}$
and $C \Pi_{\mathcal{A}^\perp} = -\Pi_{\mathcal{A}^\perp}$.
Hence, (2) is proved.
To show (3), it suffices, from (2),  to show that
$P_+ = \Pi_\mathcal{A}$
and $P_- = \Pi_{\mathcal{A}^\perp}$,
which is proved through an easy calculation.
The above argument and (3) immediately lead to (4).
}
\end{proof}
From Lemma \ref{lemmaC} and its proof,
we know that the coin operator is a unitary involution,
{\it i.e.}, $C$ is unitary and self-adjoint and that $C^2=1$.

Let $S$ be a unitary involution on $\mathcal{H}$
and set 
\[ d_B = d_A S. \]
Observe that $d_B$ is also a coisometry.
Similarly to $d_A$, 
we observe that the projection onto the closed subspace
$\mathcal{B}:= {\rm Ran}(d_B^* d_B)$
is given by $\Pi_\mathcal{B} := d_B^* d_B$ 
and 
\[ \mathcal{B}=d_B^*\mathcal{K} = {\rm Ran}(d^*_B) = {\rm Ran} \Pi_\mathcal{B}. \]
We summarize the relation between
these two coisometries $d_A$ and $d_B$
in the following:
\begin{lemma}
\label{lemmaAB}
{\rm 
Let $d_A$ and $d_B$ be as above. 
Then, we have the following:\\
(1) $d_A S = d_B$, $S d_A^* = d_B^*$.\\
(2) $d_A d_A^* = d_B d_B^* = I_\mathcal{K}$. \\
(3) $\Pi_\mathcal{A} S = S \Pi_\mathcal{B}$.\\
(4) $d_A^*$, $d_B^*$ are isometry operators and
\[ \|d_A^* f\|_\mathcal{H} = \|d_B^* f\|_\mathcal{H}  = \|f\|_\mathcal{K}
	\quad \mbox{for all $f \in \mathcal{K}$}.  \]
}
\end{lemma}
We omit the proof because it is straightforward.

\section{Abstract quantum walks and two examples}\label{sec:3}
Given a coisometry $d_A:\mathcal{H} \to \mathcal{K}$ 
and a unitary involution $S$ on $\mathcal{H}$,
we can define the coin operator $C = 2d_A^* d_A -1$ 
and the coisometry $d_B = d_A S$ as in the previous section.
Throughout this section, we fix $d_A$ and $S$
and call them a boundary operator and a shift operator, respectively.
In analogy with the twisted Szegedy walk (see Example \ref{ExampleTSW}),
we define an abstract time evolution $U$
and its discriminant $T$ as follows:
\begin{definition}
\label{Def.evo}
{\rm 
Let $d_A$, $d_B$, $C$ and $S$ be as above.
Then,
\begin{itemize}
\item[(1)] The evolution associated with 
	the boundary operator $d_A$ and the shift operator $S$ is defined by
	\[ U = S C. \] 
\item[(2)] The discriminant of $U$ is defined by
	\[ T = d_A d_B^*. \]
\end{itemize}
}
\end{definition}
$S$ and $C$ are unitary operators on $\mathcal{H}$,
so is the evolution $U$.
By definition, the discriminant $T$ 
is a self-adjoint operator on $\mathcal{K}$. 

We present the two examples. 
The first one is the extended version of the Szegedy walk on a graph;
the second one is  not directly concerned with any graph. 
\begin{example}[twisted Szegedy walk \cite{HKSS14}]
\label{ExampleTSW}
{\rm
Let $G =(V, E)$ be a (possibly infinite) graph
with the sets $V$ of vertices and $E$ of unoriented edges
($E$ can include multiple edges and self-loops). 
We use $D$ to denote the set of symmetric arcs induced by $E$.
For an arc $e\in D$, the origin and the terminus of $e$ are 
denoted by $o(e)$ and $t(e)$, respectively. 
The inverse edge of $e$ is denoted by $\bar{e}$.
Let $\mathcal{H} = \ell^2(D)$ and $\mathcal{K} = \ell^2(V)$.
A {\it boundary operator} $d_A:\mathcal{H} \to \mathcal{K}$
is defined as
\begin{align*}
(d_A\psi)(v) = \sum_{e:o(e)=v} \psi(e) \overline{w(e)}, \quad v \in V 
\end{align*}
for all $\psi \in \mathcal{H}$.
Here $w:D \to \mathbb{C}$ is a {\it weight}, satisfying 
\[ \sum_{e:o(e) = v} |w(e)|^2 = 1 \quad \mbox{for all $v \in V$}. \]
The adjoint $d_A^*:\mathcal{K} \to \mathcal{H}$ of $d_A$
is called a {\it coboundary operator} and given by
\begin{align*}
(d_A^* f)(e) = w(e) f(o(e)), \quad e \in D
\end{align*} 
for all $f \in \mathcal{K}$.
We observe that $d_A d_A^* = I_\mathcal{K}$,
because
\[ (d_Ad_A^*f)(v) = \sum_{e:o(e)=v} (d_A^*f)(e) \overline{w(e)}
	 = \sum_{e:o(e)=v}  |w(e)|^2 f(o(e)) = f(v). \]
Hence, the boundary operator $d_A$ is a coisometry.

We call a map $\theta:D \to \mathbb{C}$
a 1-form if it satisfies
\[ \theta(\bar{e}) = - \theta(e) \quad \mbox{for all $e \in D$}. \]
In \cite{HKSS14}, the twisted Szegedy walk
associated with the weight $w$ and the 1-form $\theta$
is defined as follows:
\begin{itemize}
\item[(1)] The total state space is $\mathcal{H}$;
\item[(2)] The time evolution is 
	\[ U^{(w,\theta)} = S^{(\theta)} C^{(w)}, \]
	where the coin operator $C^{(w)}$ is given by
		\[ C^{(w)} = 2 d_A^* d_A - 1 \]
	and the twisted shift operator $S^{(\theta)}$ by
		\[ (S^{(\theta)}\psi)(e) = e^{-i\theta(e)}\psi(\bar{e}), 
			\quad e \in D \]
	for all $\psi \in \mathcal{H}$;
\item[(3)] The finding probability $\nu_n(u)$ of the twisted Szegedy walk
	at time $n$ at vertex $u$ is defined by
		\[ \nu_n(u) = \sum_{e:o(e)=u}|\Psi_n(e)|^2, \]
	where $\Psi_n\in \mathcal{H}$
	is the $n$-th ($n \in \mathbb{N}$) iteration of the quantum walk 
	with the initial state $\Psi_0 \in \mathcal{H}$ ($\|\Psi_0\|^2 = 1$), 
	{\it i.e.},  $\Psi_n = (U^{(w,\theta)})^n \Psi_0$. 
\end{itemize}
Because $\theta$ is a 1-form, we know that $S^{(\theta)}$ is self-adjoint.
It is easy to check $(S^{(\theta)})^2 = 1$ by definition.
Thus, we know that $S^{(\theta)}$ is a unitary involution.

We observe that the boundary operator $d_A$, coin operator $C^{(w)}$, 
and twisted shift operator $S^{(\theta)}$ of the twisted Szegedy walk
are examples of the abstract coisometry $d_A$, coin operator $C$,
and unitary involution $S$, respectively.

Because $S^{(\theta)}$ is a unitary involution,
we know that $d_B = d_A S^{(\theta)}$, also known as a boundary operator, 
is a coisometry.
The discriminant operator on $\ell^2(V)$ 
\[ T^{(w,\theta)} = d_A d_B^* \]
is expressed by 
	\[ (T^{(w,\theta)}f)(u)=\sum_{e: t(e)=u} e^{i\theta(e)} w(e)\overline{w(\bar{e})} f(o(e)), \] 
which means that $\langle \delta_v, T^{(w,\theta)}\delta_u\rangle =0$ 
if and only if $(u,v),(v,u)\notin D$. 
This twisted version of Szegedy walk can be used effectively for 
a finite quotient graph in a crystal lattice. See \cite{HKSS14}. 

Let us set $\theta(\cdot )=0$ and $w(e)=1/\sqrt{\deg(o(e))}$, 
where $\deg (x)$ is the degree of a vertex $x$, that is, 
the number of oriented edges $e$ such that $o(e)=x$. 
Then we have
\begin{equation}
\label{RW}
(Tf)(u)=(T^{(w,\theta)})f(u)=
\sum_{e: t(e)=u}(1/\sqrt{\deg(o(e))\deg(t(e))})f(o(e)),
\end{equation}
which is unitarily equivalent to 
$P_G$
on $\ell^{2}(V,\deg )$, 
where $f\in \ell^{2}(V,\deg )$ and 
$\|f\|^{2} = \sum_{x\in V } |f(x)|^{2}\deg(x) <\infty$. 
Here 
$P_G$
is the transition probability operator of 
the symmetric 
random walk on $G$. 
We remark that $P_G=T$ 
if $G$ is $d$-regular, that is, $\deg(x) =d$ for 
any $x\in V$. 
For $w(e)=1/\sqrt{\deg(o(e))}$ and $\theta(\cdot )=0$, 
$U=U^{w,\theta}$ is said to be the evolution operator of the Grover walk:
\begin{equation}
\label{GW}
(U\psi)(e)=
\sum_{f: o(e)=t(f)}(2/\deg(o(e))-\delta_{\bar{e},f})\psi(f). 
\end{equation}
for $\psi \in \ell^{2}(D)$. 
}
\end{example}
We give an example which is not apparently related to a graph. 
\begin{example}
{\rm
{Let $\mathcal{H} = L^2(\mathbb{R}) \oplus L^2(\mathbb{R})$
and $\mathcal{K} = L^2(\mathbb{R})$.
We prepare two $C^\infty$ functions $\chi_0$ and $\chi_\infty$ satisfying $\chi^2_0(x)+\chi^2_\infty(x)=1$ for every $x\in \mathbb{R}$. 
As the boundary operator, we choose for $\psi=f\oplus g\in L^2(\mathbb{R}) \oplus L^2(\mathbb{R})$, 
	\[ d_A(f\oplus g)=\chi_0f+\chi_\infty g. \]
It is easily seen that 
	\[ d_A^* f=\chi_0f \oplus \chi_\infty f,  \]
and $d_A$ is a coisometry. Now, we choose $S$ as $S(f\oplus g)=g\oplus f$. 
Then, the unitary operator $U = S(2d_A^*d_A-1)$ is expressed by 
	\[ U(f\oplus g)=\left( 2\chi_0\chi_\infty f + (2\chi_{\infty}^2-1)g \right)\oplus \left( (2\chi_0^2-1)f + 2\chi_0\chi_\infty g \right), \]
which implies that $U$ is isomorphic to 
	\[ U\cong \begin{bmatrix}  2\chi_0\chi_\infty & 2\chi_\infty^2-1 \\ 2\chi_0^2-1 & 2\chi_0\chi_\infty \end{bmatrix}. \] 
The discriminant $T$ is equivalent to $2\chi_0\chi_\infty$.
}
}
\end{example}

\section{Invariant subspaces of $U$}
{In this section we introduce important invariant subspaces of 
abstract evolution $U$ in Definition \ref{Def.evo} to describe the spectrum. 
We list important properties of $d_A$ and $T$ in the following lemmas. }
\begin{lemma}
\label{lemmaUTSC}
{\rm
The following hold:\\
(1) $C d_A^* = d_A^*$, \quad$d_A C = d_A$.\\
(2) $C d_B^* = 2 d_A^* T -d_B^*$, \quad 
	$d_B C = 2T d_A - d_B$.\\
(3) $U d_A^* = d_B^*$, \quad$Ud_B^* = 2 d_B^* T - d_A^*$.\\
(4) $d_A U d_A^* = d_B U d_B^* = T$, \quad $d_B U d_A^* = I_{\ell^2(V)}$\\
(5) $T = d_A S d_A^* = d_A d_B^* = d_B d_A^*$.\\
(6) $d_A^* T d_A = \Pi_\mathcal{A} U \Pi_\mathcal{A}$,
	\quad $d_B^* T d_B = \Pi_\mathcal{B} U \Pi_\mathcal{B}$
}
\end{lemma}
\begin{proof}
{\rm
The proof is an easy exercise and is omitted.
}
\end{proof}
\begin{lemma}
\label{lemmaTnorm}
{\rm
$\|T\| \leq  1$.
}
\end{lemma}
\begin{proof}
The assertion follows from the following calculation:
\begin{align*}
\|T f \|^2 = \langle f, d_A S d_A^* f \rangle 
	= \langle d_A^* f, S d_A^* f \rangle 
	\leq \|d_A^* f \|^2 = \|f \|^2, \quad f \in \mathcal{K}.
\end{align*}
\end{proof}
Because $S$ is a unitary involution,
the following lemma is proved similarly to Lemma \ref{lemmaC}.
\begin{lemma}
\label{lemmaS}
\noindent {\rm
\begin{itemize}
\item[(1)] $\sigma(S) = \{ \pm 1\}$.
\item[(2)] The projection $Q_\pm$ onto $\mathcal{H}^S_\pm:= \ker(S \mp 1)$ is given by
	\[  Q_{\pm} =  \frac{1 \pm S}{2}. \]
\end{itemize}
}
\end{lemma}

\noindent
We now define three subspaces $\mathcal{L}$, $\mathcal{L}_1$, 
and $\mathcal{L}_0 \subset \mathcal{H}$ as follows:
\begin{align*}
& \mathcal{L} = \mathcal{A} + \mathcal{B}\\
& \mathcal{L}_1 = d_A^* \ker(T^2-1)^\perp + d_B^*\ker(T^2-1)^\perp \\
& \mathcal{L}_0 = d_A^* \ker(T^2-1).
\end{align*}
$\mathcal{L}_0 = \{0\}$ if and only if
$\sigma_{\rm p}(T) \cap \{ +1, -1\} = \emptyset$.
In this case, $\mathcal{L} = \mathcal{L}_1$
and thus the problem becomes simple.
We need to treat the case $\mathcal{L}_0 \not=\{0\}$ with care.
Because $\ker(T^2-1) = \ker(T-1) \oplus \ker(T+1)$,
$\mathcal{L}_0$ is decomposed into
\begin{align*} 
\mathcal{L}_0 = \mathcal{L}_0^+ \oplus \mathcal{L}_0^- , 
\end{align*}
where $\mathcal{L}_0^\pm = d_A^*\ker(T\mp 1)$.
\begin{lemma}
\label{lemmaL0}
{\rm 
\begin{itemize}
\item[(1)] $\mathcal{L}_0 = \mathcal{A} \cap \mathcal{B} = d_B^* \ker(T^2-1)$. 
\item[(2)] For all $d_A^*f^\pm \in \mathcal{L}_0^\pm$,
	$d_A^* f^\pm = \pm d_B^* f^\pm$ holds.
\item[(3)] $\mathcal{L}_0^\pm \subset \mathcal{H}_\pm^S$ and
	\[ S (d_A^* f^\pm) =  \pm d_A^* f^\pm. \]
\end{itemize}
}
\end{lemma}
\begin{proof}
{\rm
We first prove $\mathcal{L}_0 \subset \mathcal{A} \cap \mathcal{B}$.
To this end, let $d_A^* f \in \mathcal{L}_0$ ($f \in \ker(T^2-1)$) and 
$f = f^+ + f^-$ ($f^\pm \in \ker(T\mp 1)$.
Then
we observe, from Lemma \ref{lemmaAB}, that
\[ \langle d_A^* f^\pm, S d_A^*f^\pm \rangle
	= \langle f^\pm, Tf^\pm \rangle = \pm \|f^\pm\|^2 
		= \pm \|d_A^* f^\pm\|^2. \] 
Hence, we have
\[ \langle d_A^*f^\pm, (S \mp 1) d_A^* f^\pm \rangle = 0. \] 
By Lemma \ref{lemmaS},
we have
\[ \| Q_\mp d_A^*f^\pm \| = 0. \]
Noting that $Q_\mp = 1- Q_\pm$,
we obtain
\[ d_A^* f^\pm = Q_\pm d_A^* f^\pm. \]
Hence, by Lemma \ref{lemmaAB}, 
\begin{equation} 
\label{*01}
d_B^* f^\pm = S d_A^* f^\pm 
	= S Q_\pm d_A^* f^\pm = \pm d_A^* f^\pm.
\end{equation} 
Thus, we have
\begin{equation}
\label{032611:11} 
d_A^* f = d_A^* (f^+ + f^-) 
	= d_B^* ( f^+ - f^-) \in \mathcal{A} \cap \mathcal{B}.
\end{equation}

We prove the converse statement.
To this end, let $\psi \in \mathcal{A} \cap \mathcal{B}$.
This can be represented in two ways:
\[ \psi = d_A^* f = d_B^* g, \quad f, g \in \mathcal{K}. \]
By Lemma \ref{lemmaAB} and Lemma \ref{lemmaUTSC},
we have
\[ f = T g, \quad g = T f, \]
which imply
\[ f = T^2 f, \quad g = T^2 g. \]
Thus, we know that $f, g \in \ker(T^2-1)$.
In particular, we have  $\psi \in \mathcal{L}_0$, 
and the converse statement 
$\mathcal{A} \cap \mathcal{B} \subset \mathcal{L}_0$ is proved.
Hence, we have $\mathcal{L}_0 = \mathcal{A} \cap \mathcal{B}$.
Moreover, from \eqref{032611:11},
we also have $\mathcal{L}_0 
	\subset d_B^* \ker (T^2-1)$.
In a way similar to the above,
we can show $d_B^*\ker (T^2-1) \subset \mathcal{A} \cap \mathcal{B}$.
Thus, (1) is proved.

\eqref{*01} implies (2) and (3). 
}
\end{proof}
\begin{lemma}
{\rm
\[ \mathcal{L} = \mathcal{L}_1 \oplus \mathcal{L}_0. \]
Moreover, for all $\psi \in \mathcal{L}$,
there exist unique $ f, g \in \ker(T^2 -1)^\perp$ and $h_0 \in \ker(T^2-1)$
such that
\begin{align} 
\label{decomp01}
\psi = d_A^* f + d_B^* g + d_A^* h_0. 
\end{align}
}
\end{lemma}
\begin{proof}
{\rm
We first prove that $\mathcal{L}_1 \perp \mathcal{L}_0$.
To this end, let $\psi_1 \in \mathcal{L}_1$ and $\psi_0 \in \mathcal{L}_0$.
Then $\psi_0$ and $\psi_1$ can be represented as
\begin{align*} 
& \psi_1 = d_A^* f + d_B^* g, \quad f,g \in \ker(T^2-1)^\perp, \\
& \psi_0 = d_A^* h_0, \quad h_0 \in \ker(T^2-1)
\end{align*}
By the decomposition $h_0 = h_0^+ + h_0^-$ ($h_0^\pm \in \ker(T\mp 1)$)
and Lemma \ref{lemmaL0}, 
we have
\begin{align*}
\langle \psi_1, d_A^* h_0^\pm \rangle
	& = \langle  d_A^* f + d_B^* g, d_A^* h_0^\pm \rangle
		=  \langle  d_A^* f, d_A^* h_0^\pm \rangle
			\pm  \langle  d_B^* g,  d_B^* h_0^\pm \rangle, \\
		& = \langle f \pm g, h_0^\pm \rangle = 0.
\end{align*}
Because $\psi_0 = d_A^* h^+ + d_A^* h^-$,
we obtain $\langle \psi_1, \psi_0 \rangle = 0$.
Hence, we have the desired result.

To prove \eqref{decomp01},
let $\psi \in \mathcal{L}$.
Then there exist $\tilde f, \tilde g \in \mathcal{K}$ such that
\[ \psi = d_A^* \tilde f + d_B^* \tilde g. \]
Decomposing $\tilde{f}$ and $\tilde{g}$ as
$\tilde f = f + f_0$, $\tilde g = g + g_0$ ($f,g \in \ker(T^2-1)^\perp$,
$f_0, g_0 \in \ker(T^2-1)$), 
we have
\[ \psi = d_A^* f + d_B^* g + d_A^* f_0 + d_B^* g_0. \]
Because $d_B^* g_0 \in \mathcal{L}_0$, 
decomposing $g_0$ as
$g_0 = g_0^+ + g_0^-$ ($g_0^\pm \in \ker(T\mp 1)$),
and using Lemma \ref{lemmaL0},
we obtain 
\[ d_B^* g_0 = d_A^* (g_0^+ - g_0^-). \]
Letting $h_0 = f_0  + g_0^+ - g_0^-$,
we have the decomposition \eqref{decomp01},
which implies $\mathcal{L} \subset \mathcal{L}_1 \oplus \mathcal{L}_0$.
Since the converse inclusion is clear,
we obtain $\mathcal{L} = \mathcal{L}_1 \oplus \mathcal{L}_0$.

We prove the uniqueness of the decomposition \eqref{decomp01}.
We assume that $\psi \in \mathcal{L}$
can be represented in two ways:
\begin{align*} 
\psi 
	& = d_A^* f + d_B^* g + d_A^* h_0 \\
	& =d_A^* f^\prime + d_B^* g^\prime + d_A^* h_0^\prime, 
\end{align*}
where $f,f^\prime, g, g^\prime \in \ker(T^2 -1)^\perp$
and $h_0, h_0^\prime \in \ker(T^2-1)$.
Then we have
\[ d_A^*(f -f^\prime + h_0 - h_0^\prime) 
	= d_B^* (g^\prime-g) \in \mathcal{A} \cap \mathcal{B}. \]
This implies $g^\prime-g \in \ker(T^2-1) \cap \ker(T^2-1)^\perp$
and hence $g^\prime=g$.
Moreover, we then know that
\[ d_A^*(f -f^\prime + h_0 - h_0^\prime)=0 \]
and hence that
\[ f- f^\prime = h_0^\prime-h_0 \in \ker(T^2-1)^\perp \cap \ker(T^2-1). \]
Thus we have $f^\prime=f$, $h_0^\prime = h_0$.
Hence, uniqueness is proved.
}
\end{proof}
\begin{lemma}
\label{lemmadAdB}
{\rm
\[ \mathcal{L}^\perp = \ker (d_A) \cap \ker (d_B). \]
}
\end{lemma}
\begin{proof}
{\rm
We first prove that $\mathcal{L}^\perp \subset \ker(d_A) \cap \ker(d_B)$.
Let $\psi \in \mathcal{L}^\perp$.
Then
we have for all $d_A^* f + d_B^* g \in \mathcal{L}$
($f, g \in \mathcal{K}$)
\[ \langle \psi, d_A^* f \rangle = - \langle \psi, d_B^* g \rangle. \]
Let $g=0$ (resp.  $f=0$).
We have $\langle d_A\psi, f \rangle = 0$
for all $f \in \mathcal{K}$
(resp. $\langle d_B\psi, g\rangle = 0$ for all $g \in \mathcal{K}$).
Hence, we obtain $\psi \in \ker(d_A) \cap \ker(d_B)$.

Conversely, let $\psi \in \ker(d_A) \cap \ker(d_B)$.
We have
\[ \langle \psi, d_A^* f + d_B^* g \rangle = 0, \quad f, g \in \mathcal{K}. \]
Hence, we obtain $\psi \in \mathcal{L}^\perp$
and $\ker(d_A) \cap \ker(d_B) \subset \mathcal{L}^\perp$ is proved.
}
\end{proof}

\begin{proposition}
\label{propinv}
{\rm
$\overline{\mathcal{L}_1}$, $\mathcal{L}_0$ and $\mathcal{L}^\perp$
are invariant subspaces of $U$, {\it i.e.},
\[ U \mathcal{V} \subset \mathcal{V} \]
for $\mathcal{V} = \overline{\mathcal{L}_1}, \ \mathcal{L}_0$ 
and $\mathcal{L}^\perp$.
}
\end{proposition}
\begin{proof}
{\rm
We first prove that
$U \overline{\mathcal{L}_1} \subset \overline{\mathcal{L}_1}$.
It suffices to show that $U \mathcal{L}_1 \subset \mathcal{L}_1$.
To this end,
let $\psi \in \mathcal{L}_1$ and write
\[ \psi = d_A^* f + d_B^* g,
	\quad f \in \ker(T^2-1)^\perp, \ g \in \ker(T^2-1)^\perp. \]
By Lemma \ref{lemmaUTSC},
we know that
\begin{align*}
U \psi  = Ud_A^* f + Ud_B^* g 
	= d_B^* (f + 2Tg) -d_A^*g.
\end{align*}
Because $f + 2Tg \in \ker(T^2-1)^\perp$,
we have $U\psi \in \mathcal{L}_1$.
Hence $U$ leaves $\overline{\mathcal{L}_1}$ invariant.

We next prove that $U \mathcal{L}_0 \subset \mathcal{L}_0$.
Let $\psi = d_A^* (h_0^+ + h_0^-) \in \mathcal{L}_0$
($h_0^\pm \in \ker(T\mp1)$).
Then, by Lemma \ref{lemmaL0}, we have
\[ U\psi = d_B^* h_0^+ + d_B^* h_0^- = d_A^*(h_0^+ - h_0^-) \in \mathcal{L}_0. \]
Hence, we obtain the desired result.

We prove that $U  \mathcal{L}^\perp \subset \mathcal{L}^\perp$.
Combining Lemma \ref{lemmadAdB} with $d_B = d_A S$,
we know that
\begin{equation} 
\label{Lperp}
\mathcal{L}^\perp = \{ \psi \in \ell^2(D) \mid d_A \psi = 0, \ d_A (S \psi) = 0 \}.
\end{equation}
Since, by Lemma \ref{lemmaC}, 
we have
$\mathcal{L}^\perp \subset \ker (d_A) = {\rm Ran} P_-$,
we know that $C\psi = - \psi$ holds for all $\psi \in \mathcal{L}^\perp$.
Hence we have
\begin{equation} 
\label{ULperp}
U \psi = - S\psi, \quad \mbox{for all $\psi \in \mathcal{L}^\perp $}.
\end{equation}
We observe from \eqref{Lperp},
that $U\psi \in \mathcal{L}^\perp$.
This completes the proof.
}
\end{proof}

Proposition \ref{propinv} implies that
$U$ is reduced by the subspaces
$\overline{\mathcal{L}_1}, \ \mathcal{L}_0$ and $\mathcal{L}^\perp$
and is decomposed into
\[ U = U_{\overline{\mathcal{L}_1}} 
		\oplus U_{\mathcal{L}_0}  \oplus U_{\mathcal{L}^\perp}
, \]
where we have used $U_\mathcal{V}$ to denote the restriction of $U$
to a subspace $\mathcal{V}$.
Then we have
\begin{align}
\label{decomp02} 
& \sigma (U) = \sigma( U_{\overline{\mathcal{L}_1}}  ) \cup
	\sigma(U_{\mathcal{L}_0} ) \cup \sigma(U_{\mathcal{L}^\perp}), \\
\label{decomp03} 
& \sigma_{\rm p} (U) = \sigma_{\rm p}( U_{\overline{\mathcal{L}_1}}  ) \cup
	\sigma_{\rm p}(U_{\mathcal{L}_0} ) \cup \sigma_{\rm p}(U_{\mathcal{L}^\perp}).
\end{align}

\section{Eigenvalues of $U$}	
\subsection{Eigenspaces and invariant subspaces}
In this subsection,
we prove the following proposition:
\begin{proposition}
\label{prop*00}
{\rm
The following hold:
\begin{itemize}
\item[(1)] $\sigma_{\rm p}(U)\cap \{ \pm 1\} 
	\subset \sigma_{\rm p}(U_{{\mathcal{L}_0}}) \cup \sigma_{\rm p}(U_{\mathcal{L}^\perp})$.
\item[(2)] $\sigma_{\rm p}(U_{\mathcal{L}_0}) \subset \{\pm 1\}$.
\item[(3)] $\sigma_{\rm p}(U) \setminus \{\pm 1\} 
	= \sigma_{\rm p}(U_{\overline{\mathcal{L}_1}})
		\subset \{ e^{i\xi} \mid  \cos \xi \in \sigma_{\rm p}(T),
	 	\xi \in (0,\pi) \cup (\pi, 2\pi) \}$.
\end{itemize}
}
\end{proposition}
Let $\lambda \in \mathbb{C}$ be an eigenvalue of $U$
and $\psi \in \mathcal{H}$ its eigenvector:
\begin{equation}
\label{*002*} 
U \psi = \lambda \psi, \quad \psi \not=0. 
\end{equation}
Because $U$ is unitary,
$|\lambda| = 1$.
Using the decomposition (4) of Lemma \ref{lemmaC},
we can write 
\begin{equation}
\label{*001*} 
\psi = d_A^* f + \psi_0, \quad f \in \mathcal{K}, \ \psi_0 \in \ker d_A. 
\end{equation}

\begin{lemma}
{\rm
Let $f, \psi_0$ and $\lambda$ be as above.
Then,
\begin{align}
\label{**01}
& (T-\lambda)f 
	= d_B \psi_0, \\
\label{**02}
& (\bar\lambda -  T)f 
	= d_B\psi_0.
\end{align}
}
\end{lemma}
\begin{proof}
Because $\ker d_A = {\rm Ran} P_-$ from Lemma \ref{lemmaC}, 
we have $C\psi_0 = - \psi_0$.
Substituting \eqref{*001*} into \eqref{*002*},
we have
\begin{equation}
\label{*1227} 
U\psi = U(d_A^*f + \psi_0) = d_B^*f - S\psi_0.
\end{equation}
Hence, it holds, from \eqref{*002*}, that
\begin{equation}
\label{*02}
(d_B^*- \lambda d_A^*) f = (S + \lambda) \psi_0.
\end{equation}
%
%
Letting $d_A$ and $d_B$ act on
\eqref{*02}, 
we obtain 
\begin{align*}
& (T-\lambda)f 
	= d_A (S + \lambda) \psi_0 
	= d_B \psi_0, \\
& (1 - \lambda T)f 
	= d_B(S + \lambda) \psi_0
	= \lambda d_B\psi_0. 
\end{align*}
Noting that $\bar{\lambda} = \lambda^{-1}$ holds
from $|\lambda|=1$,
we have the desired result.
\end{proof}
\begin{proof}[Proof of Proposition \ref{prop*00}]
Let $f, \psi_0$ and $\lambda$ be as above.
Combining \eqref{**01} with \eqref{**02},
we have $(T- {\rm Re}\lambda) f = 0$, 
and hence
\[ f \in \ker (T -{\rm Re}\lambda). \]

We first consider the case in which $\lambda = \pm 1$.
Then, we have $f \in \ker(T \mp 1)$.
Hence, by \eqref{**01}, 
we obtain $d_B \psi_0 = 0$ and
\[ \psi_0 \in \ker (d_A) \cap \ker (d_B). \]
Because $\ker(T \mp 1) \subset \ker (T^2-1)$,
we get
\[ \psi = d_A^* f + \psi_0 \in \mathcal{L}_0 \oplus \mathcal{L}^\perp 
	= \mathcal{L}_1^\perp. \]
If $f\not=0$, 
then $f \in \ker(T\mp1)$ is an eigenvector of $T$.
If $f=0$, then $\psi_0 \not=0$, because $\psi \not=0$.
By \eqref{*02}, we have $S\psi = -\lambda \psi$:
therefore, we observe from \eqref{*1227} that $U\psi_0  = \lambda \psi_0$.
Hence we know that
\[ \sigma_{\rm p}(U)\cap \{ \pm 1\} 
	\subset \sigma_{\rm p}(U_{\mathcal{L}_0}) \cup \sigma_{\rm p}(U_{\mathcal{L}^\perp}),
	\quad
		\sigma_{\rm p} (U_{\overline{\mathcal{L}_1}})
	\cap \{\pm 1\} = \emptyset. \]

Let us next consider the case where 
$\lambda = e^{i\xi}$ ($\xi \in (0,\pi) \cup (\pi, 2\pi)$).
Then it holds that
\[  f \in \ker (T-\cos \xi) \subset \ker(T^2-1)^\perp.
	\]
Since  $\lambda \not= \pm 1$, we observe that
$S+\lambda$ has a bounded inverse with
\begin{equation} 
\label{327946}
(S + \lambda)^{-1} = \frac{S-\lambda}{1-\lambda^2}. 
\end{equation}
Hence, by \eqref{*02}, we have
\[ \psi_0 = (S+\lambda)^{-1}(d_B^* - \lambda d_A^*) f
	= \frac{1}{1-\lambda^2} \left( (1+\lambda^2) d_A^* -2\lambda d_B^* \right)f. \]
Thus we obtain
\begin{equation*} 
\psi = d_A^* f + \psi_0
	= \frac{2}{1-\lambda^2} \left(d_A^* -\lambda d_B^* \right)f \in \mathcal{L}_1,
\end{equation*}
which implies $f\not=0$, because $\psi \not=0$.
Therefore we know that
$f$ is an eigenvector of $T$ corresponding to ${\rm Re}\lambda = \cos \xi$
and
\[ \sigma_{\rm p}(U) \setminus \{\pm 1\} 
	= \sigma_{\rm p}(U_{\overline{\mathcal{L}_1}}). \]
\end{proof}

\subsection{Eigenvalues of $U_{\mathcal{L}_0}$}
Let $m_{\pm} = {\rm dim}\ker (T \mp 1)$. 
We use $\{ \pm 1\}^{m_\pm}$ to denote multiplicity if $m_\pm > 0$
and use the convention that $\{ \pm 1\}^{m_\pm}= \emptyset$ if $m_\pm = 0$.
Our purpose in this subsection is to prove
\begin{proposition}
\label{prop*01}
{\rm
The following hold:
\begin{itemize}
\item[(1)] $U_{\mathcal{L}_0} = I_{\mathcal{L}_0^+} \oplus (- I_{\mathcal{L}_0^-})$.
\item[(2)] $\sigma(U_{\mathcal{L}_0}) =  \sigma_{\rm p}(U_{\mathcal{L}_0}) 
	= \{1\}^{m_+} \cup \{-1\}^{m_{-}}$.
\end{itemize}
}
\end{proposition}

We need the following lemma.
\begin{lemma}
\label{lemmaisom}
{\rm
$d_A^*\mid_{\ker(T \mp 1)}$
is a bijection with the inverse 
\[ (d_A^*\mid_{\ker(T \mp 1)})^{-1} = d_A\mid_{\mathcal{L}_0^\pm}. \]
}
\end{lemma}
\begin{proof}
{\rm
It suffices to show that
\[ d_A^*\mid_{\ker(T \mp 1)}d_A\mid_{\mathcal{L}_0^\pm} = I_{\mathcal{L}_0^\pm},
	\quad d_A\mid_{\mathcal{L}_0^\pm} d_A^*\mid_{\ker(T \mp 1)} = I_{\ker(T \mp 1)}. \]
Let  $d_A^* h_0^\pm \in \mathcal{L}_0^\pm$.
Then, we have
$d_A (d_A^* h_0^\pm) = h_0^\pm \in \ker (T \mp 1)$ and
\[ d_A^*(d_A (d_A^* h_0^\pm) ) = d_A^* h_0^\pm, \]
which implies the former equation.

Conversely,
for all $h_0^\pm \in \ker (T \mp 1)$,
we have
$d_A^* h_0^\pm \in \mathcal{L}_0^\pm$ and
\[ d_A(d_A^* h_0^\pm) = h_0^\pm,  \]
which implies the latter equation.
}
\end{proof}
\begin{proof}[Proof of Proposition \ref{prop*01}]
Using Lemma \ref{lemmaUTSC} and Lemma \ref{lemmaL0},
we obtain the following: 
for $d_A^* h_0^\pm \in \mathcal{L}_0^\pm$,
\begin{equation}
\label{eqUSpm1}
U (d_A^* h_0^\pm) = S (d_A^* h_0^\pm) = \pm d_A^* h_0^\pm,
\end{equation}
which implies that $U_{\mathcal{L}_0}$ leaves $\mathcal{L}_0^\pm$ invariant
and (1) holds. 
By Lemma \ref{lemmaisom},
we have ${\rm dim} \ker (T \mp 1) = {\rm dim} \mathcal{L}_0^\pm$;
therefore, (1) leads to (2).
This completes the proof.
\end{proof}

\subsection{Eigenvalues of $U_{\mathcal{L}^\perp}$}
Let $\mathcal{L}^\perp_\pm = \mathcal{L}^\perp \cap \mathcal{H}^S_\mp$
and $M_{\pm} = {\rm dim }  \mathcal{L}^\perp_\pm$.
We prove the following:
\begin{proposition}
\label{prop*02}
{\rm
The following hold:
\begin{itemize}
\item[(1)]  $U_{\mathcal{L}^\perp} = I_{\mathcal{L}^\perp_+} 
	\oplus (- I_{\mathcal{L}^\perp_-})$.
\item[(2)] $\sigma(U_{\mathcal{L}^\perp}) = \sigma_{\rm p}(U_{\mathcal{L}^\perp})
	= \{1 \}^{M_+} \cup \{-1\}^{M_{-}} $.
\end{itemize}
}
\end{proposition}
\begin{proof}
We first prove that $\mathcal{L}^\perp$ can be decomposed into
$\mathcal{L}^\perp
	= \mathcal{L}^\perp_+ \oplus \mathcal{L}^\perp_-$.
To this end, let us write $\psi \in \mathcal{L}^\perp$ as
\[ \psi = \psi_+ + \psi_-, \quad \psi_\pm := Q_\mp \psi. \]
Because, by \eqref{Lperp}, we have $d_A \psi = 0$ and $d_A(S\psi) = 0$,
we know that
\begin{align*} 
& d_A \psi_\pm = d_A \left( \frac{1\mp S}{2} \psi \right) = 0, \\
& d_A (S\psi_\pm) =  d_A \left( \frac{ S\mp 1}{2} \psi \right) = 0.
\end{align*}
Hence, by \eqref{Lperp} again,
we have $\psi_\pm \in \mathcal{L}^\perp_\pm$.
Because, by definition,
$\mathcal{L}^\perp_+ \perp \mathcal{L}^\perp_-$,
we know that $\mathcal{L}^\perp
	= \mathcal{L}^\perp_+ \oplus \mathcal{L}^\perp_-$.
	
We next prove that $\mathcal{L}^\perp_\pm$ is an invariant subspace of $U$.
To this end, let $\psi_\pm \in \mathcal{L}^\perp_\pm$.
By \eqref{ULperp}, we know that
\begin{equation*} 
U \psi_\pm = - S \psi_\pm =  \pm \psi_\pm,
\end{equation*}
where we have used $\psi_\pm \in \mathcal{H}_\mp^S$ in the last equality.
The above equation implies that
$U\mathcal{L}^\perp_\pm \subset \mathcal{L}^\perp_\pm$ and (1).
(2) immediately follows from (1).
\end{proof}

\subsection{Eigenvalue of $U_{\overline{\mathcal{L}_1}}$}
In this section,
we prove:
\begin{proposition}
\label{prop*03}
{\rm
\begin{itemize}
\item[(1)] $\sigma_{\rm p}(U_{\overline{\mathcal{L}_1}}) 
		= \{ e^{i\xi} \mid  \cos \xi \in \sigma_{\rm p}(T),
	 	\xi \in (0,\pi) \cup (\pi, 2\pi) \}$.
\item[(2)] For all $\xi  \in (0,\pi) \cup (\pi, 2\pi)$,
	it holds that
	\[ {\rm dim} \ker (U - e^{i\xi}) = {\rm dim} \ker (T -\cos \xi). \]
\end{itemize}
}
\end{proposition}
Summarizing  Propositions \ref{prop*00}, \ref{prop*01}, 
\ref{prop*02} and \ref{prop*03},
we obtain the following:
\begin{theorem}
\label{Mainthm01}
{\rm
The set of eigenvalues of $U$ is given by
\[ \sigma_{\rm p}(U) 
	= \{ e^{i\xi} \mid \cos \xi \in \sigma_{\rm p}(T), \ \xi \in [0,2\pi) \}
		\cup \{+1\}^{M_+} \cup \{-1\}^{M_-} \]
and the multiplicity is given by
\[ {\rm dim} \ker(U - e^{i\xi}) 
	= \begin{cases}
		{\rm dim} \ker (T - \cos \xi), & \xi \in (0, \pi) \cup (\pi, 2\pi), \\
		M_+ + m_+, & \xi = 0, \\
		M_- + m_-, & \xi = \pi, 
		\end{cases}
\]
where 
	\begin{equation} \label{Mpm}
        M_{\pm} = {\rm dim}\mathcal{L}^\perp_\pm \mathrm{\;\;and\;\;} m_\pm = {\rm dim}\ker(T\mp 1).
        \end{equation}
}
\end{theorem}
The following corollary is immediately obtained from Theorem \ref{Mainthm01}.
\begin{corollary}
{\rm
Let $M_\pm$ and $m_\pm$ be as above.
Then:
\begin{itemize}
\item[(1)] $\sigma_{\rm p}(U) \setminus \{\pm 1\} 
			=  \{ e^{i\xi} \mid 
				\cos \xi \in \sigma_{\rm p}(T) \setminus \{\pm1\}, \ \xi \in [0,2\pi) \}$;
\item[(2)] $\sigma_{\rm p}(U) \cap \{\pm 1\} 
			= \{+1 \}^{M_+ + m_+} \cup \{-1\}^{M_-+m_-}$.
\end{itemize} 
}
\end{corollary}
\begin{proof}[Proof of Proposition \ref{prop*03}]
We first prove (i). 
Because we have already proved in  Proposition \ref{prop*00} that
$\sigma_{\rm p}(U_{\overline{\mathcal{L}_1}}) 
		\subset \{ e^{i\xi} \mid  \cos \xi \in \sigma_{\rm p}(T),
	 	\xi \in (0,\pi) \cup (\pi, 2\pi) \}$,
we need only to prove the converse statement.
To this end,
it suffices to show $e^{i \xi} \in \sigma_{\rm p}(U_{\overline{\mathcal{L}_1}})$
for $\cos \xi \in \sigma_{\rm p}(T) \setminus \{\pm 1\}$.
Let $f \in \ker(T-\cos \xi) \setminus\{0\}$ be an eigenvector of $T$
with eigenvalue $\lambda = e^{i\xi}$
and set
\begin{equation}
\label{327*100*} 
\psi = \left(d_A^* -\lambda d_B^* \right)f,
\end{equation}
which is clearly in $\mathcal{L}_1$.
We observe that $\psi \not=0$,
because we know that $\psi = (1-\lambda S) d_A^* f$
and, from \eqref{327946}, that
$(1-\lambda S)^{-1} = -\bar\lambda (S+\bar\lambda)^{-1}$ is bounded.
Because $Tf = \cos \xi f$, we have
\begin{align*}
U \psi = d_B^* (1- 2\lambda T )f + \lambda d_A^* f
	= \lambda \psi.
\end{align*}
Hence we have the desired result
and (1) is proved.

To prove (ii), we consider the multiplicity of  
$e^{i\xi}$ ($\cos \xi \in \sigma_{\rm p}(T) \setminus \{\pm 1\}$).
Let $\lambda = e^{i\xi} \in \sigma_{\rm p}(U_{\overline{\mathcal{L}_1}})$
and $\psi$ be its eigenvector.
Then, from the argument in the proof of Proposition \ref{prop*00},
we know that $\psi$ is of the form \eqref{327*100*}
up to a constant factor.
As is shown in the proof of (i),
we know that, if $\psi$ is of the form \eqref{327*100*},
then $\psi \in \ker(U-\lambda)$.
Therefore we have
\[ \ker(U-\lambda) = \{ \psi 
	= (d_A^* -\lambda d_B^*)f \mid f \in \ker(T-\cos \xi) \}. \]
Let us now define a map $K_\lambda: \ker(T-\cos\xi) \to \ker(U -e^{i\xi})$
by
\[ K_\lambda = d_A^* - \lambda d_B^* = (1-\lambda S)d_A^*.  \]
Then, $K_\lambda$ is a surjection,
because $\ker(U-\lambda) = K_\lambda \ker(T-\cos \xi)$.
We also observe that an operator
\[ M_\lambda = \frac{\lambda}{1-\lambda^2} d_A (S + \bar{\lambda}) \]
satisfies
\[ M_\lambda K_\lambda = 1. \]
Thus, we know that $K_\lambda$ is a bijection
and obtain the desired result.
\end{proof}
\begin{remark}
{\rm
From the above proof,
we know that,
for $\lambda = e^{i\xi} \not=\pm 1$,
\begin{itemize}
\item[(1)] $\ker(U-\lambda) = K_\lambda \ker(T-\cos \xi)$,
\item[(2)] $\ker(T-\cos \xi) = M_\lambda \ker(U-\lambda)$. 
\end{itemize}
}
\end{remark}

\section{Spectra of $U$}
In this section,
we characterize the spectrum $\sigma(U)$.
\begin{proposition}
\label{prop100}
{\rm
It holds that
\[ \sigma(U_{\overline{\mathcal{L}_1}}) \setminus K
	= \{ e^{i\xi} \mid \cos \xi \in \sigma(T), \ \xi \in [0,2\pi) \} 
		\setminus K,  \]
where $K := \{+1, -1\} \cap \sigma_{\rm p}(T)$.
In particular, if $K = \emptyset$, then
\[ \sigma(U_{\overline{\mathcal{L}_1}}) 
	= \{ e^{i\xi} \mid \cos \xi \in \sigma(T), \ \xi \in [0,2\pi) \}. \]
}
\end{proposition}
Before proving this proposition,
we first state the following:
\begin{theorem}
\label{Mainthm02}
{\rm
$\sigma(U) 
	= \{ e^{i\xi} \mid \cos \xi \in \sigma(T), \ \xi \in [0,2\pi) \} 
		\cup \{1\}^{M_+} \cup \{-1\}^{M_-}$.
}
\end{theorem}

\begin{proof}
{\rm
Combining \eqref{decomp02}
with Propositions \ref{prop*01} and \ref{prop*02},
we have
\begin{align*} 
\sigma(U) 
	& = \sigma(U_{\overline{\mathcal{L}_1}}) 
			\cup \sigma(U_{\mathcal{L}_0}) 
			\cup \sigma(U_{\mathcal{L}^\perp}) \\
	& = \sigma(U_{\overline{\mathcal{L}_1}}) 
			\cup \sigma_{\rm p}(U_{\mathcal{L}_0}) 
			\cup \sigma_{\rm p}(U_{\mathcal{L}^\perp}).
\end{align*}
Noting that $\sigma_{\rm p}(U_{\mathcal{L}_0}) = K$,
we observe from Proposition \ref{prop100} that
\[ \sigma(U_{\overline{\mathcal{L}_1}}) \cup \sigma_{\rm p}(U_{\mathcal{L}_0})
	= \{ e^{i\xi} \mid \cos \xi \in \sigma(T), \ \xi \in [0,2\pi) \}. \]
Since $\sigma_{\rm p}(U_{\mathcal{L}^\perp}) = \{+1\}^{M_+} \cup \{-1\}^{M_-}$,
the theorem is proved.
%
}
\end{proof}

Proposition \ref{prop100} is immediately proved by the following lemma:
\begin{lemma}
\label{lemma*11*} 
{\rm
The following hold:
\begin{itemize}
\item[(i)]
$\sigma(U_{\overline{\mathcal{L}_1}}) 
	\subset \{ e^{i\xi} \mid \cos \xi \in \sigma(T), \ \xi \in [0,2\pi) \}
$;
\item[(ii)]
$\{ e^{i\xi} \mid \cos \xi \in \sigma(T), \ \xi \in [0, 2\pi) \} \setminus K 
	\subset \sigma(U_{\overline{\mathcal{L}_1}})$.
\end{itemize}
}
\end{lemma}

\begin{proof}[Proof of Lemma \ref{lemma*11*}] 
{\rm
(i)
Assume that $e^{i\xi} \in \sigma(U_{\overline{\mathcal{L}_1}})$.
Then, from the fact that $U_{\overline{\mathcal{L}_1}}$ is unitary,
we know that 
there exists a sequence $\{\psi_n\}$ of normalized vectors
such that 
\[ \lim_{n\to \infty}\|(U_{\overline{\mathcal{L}_1}}-e^{i\xi})\psi_n\|=0. \]
Let $f_n = d_A \psi_n$.
Assume that $\lim_{n \to \infty}f_n = 0$, which implies 
$\lim_{n \to \infty} d_A^*d_A \psi_n = \lim_{n \to \infty} d_A^* f_n = 0$
and hence
\begin{align*}
\langle U\psi_n, e^{i\xi} \psi_n \rangle
= e^{i\xi} \langle S(2d_A^*d_A - 1) \psi_n, \psi_n \rangle
= -e^{i\xi} \langle S \psi_n, \psi_n \rangle + o(1).
\end{align*}
By the definition of $\psi_n$, we have
\[ \langle U\psi_n, e^{i\xi} \psi_n \rangle 
	= \langle U\psi_n, U\psi_n \rangle + \langle U\psi_n, (e^{i\xi} - U) \psi_n \rangle
		= 1 + o(1). \]
Combining the above two equations,
we obtain
\begin{equation}
\label{32801} 
\lim_{n \to \infty} \langle S \psi_n, \psi_n \rangle= - e^{-i\xi}. 
\end{equation}
Because $S$ is self-adjoint,
\eqref{32801} is allowed only when $\xi = 0, \pi$. 

Let us first consider the case in which $\xi \in (0, \pi) \cup (\pi, 2\pi)$.
Then, $f_n$ does not converge to zero,
because, from the above argument,
\eqref{32801} contradicts $\lim_{n \to \infty}f_n = 0$. 
Hence, there exists a subsequence $\{f_{n_k}\}$ 
such that $\inf_k \|f_{n_k}\| = :c >0$ holds.
We write $f_{n_k} = d_A \psi_{n_k}$ simply as $f_k = d_A \psi_k$.
Then, we observe that
\begin{align}
T f_k & = d_B (d_A^* d_A) \psi_k
	= d_A S \left( \frac{C + 1}{2}\right) \psi_k
	= d_A \left( \frac{U + S}{2}\right)\psi_k \notag \\
	& = \frac{1}{2}  d_A ( e^{i\xi} + S )\psi_k + o(1). \label{*100*}
\end{align}
We also observe that
\begin{align}
S\psi_k & = e^{-i\xi} S(e^{i\xi}\psi_k) =  e^{-i\xi} S(U\psi_k) + o(1) \notag \\
	& =  e^{-i\xi} C\psi_k + o(1). \label{*200*}
\end{align}
Combining \eqref{*200*} with \eqref{*100*},
and using the fact that $d_A C = d_A$,
we obtain
\begin{align*}
T f_k 
	& = \frac{1}{2}  d_A ( e^{i\xi} + e^{-i\xi} C )\psi_k + o(1) \\
	& = (\cos \xi) f_k + o(1).
\end{align*}
Let $\tilde{f}_k := f_k/\|f_k\|$. 
Then, we know that $\|\tilde{f}_k\|=1$, and that
\begin{align*}
\|(T - \cos \xi )\tilde{f}_k\| \leq \frac{1}{c} \|(T - \cos \xi) f_k\| = o(1),
\end{align*}
where $c=\inf_k\|f_k\|>0$.
Thus, we obtain
\[ \sigma(U_{\overline{\mathcal{L}_1}}) \setminus\{\pm 1\} 
	\subset \{ e^{i\xi} \mid \cos \xi \in \sigma(T), 
			\ \xi \in (0,\pi) \cup (\pi,2\pi) \}. \]

We next consider the case in which $\xi = 0, \pi$,
{\it i.e.}, $\pm 1 \in  \sigma(U_{\overline{\mathcal{L}_1}})$.
In this case, assuming that $f_n = d_A \psi_n$ satisfies $\lim_{n \to \infty} f_n=0$,
we have
\begin{equation} 
\label{*300*}
\lim_{n \to \infty} \langle S\psi_n, \psi_n \rangle = \mp 1. 
\end{equation}
Using \eqref{*200*} with $\xi=0,\pi$,
we have
\begin{align*} 
S\psi_n & = \mp C \psi_n + o(1) \\
	& = \mp (2d_A^*d_A-1) \psi_n + o(1) \\
	& = \pm \psi_n + o(1).
\end{align*}
Substituting this equation into 
the left-hand side of \eqref{*300*},
we obtain
\[ \lim_{n \to \infty} \langle S\psi_n, \psi_n \rangle
	= \pm 1, \]
which contradicts \eqref{*300*}.
Hence we know that $f_n$ does not converge to zero.
Thus, from the same argument as above, 
we obtain $\pm 1 \in \sigma(T)$.
Therefore 
(i) is proved.

(ii)
We write
\[ \{ e^{i\xi} \mid \cos \xi \in \sigma(T), \ \xi \in [0, 2\pi) \} \setminus K 
	= I_1 \cup I_2, \]
where 
\begin{align*} 
& I_1 := \{ e^{i\xi} \mid \cos \xi \in \sigma(T), \ \xi \in (0, \pi) \cup (\pi, 2\pi) \}, \\
& I_2 := \sigma_{\rm c}(T) \cap \{+1,-1\}.
\end{align*}
Therefore, it suffices to show that $I_i \subset \sigma(U_{\overline{\mathcal{L}_1}})$
 ($i=1,2$).
 
Assume that $e^{i\xi} \in I_1$.
Then we know that $\cos \xi \in \sigma(T) \cap (-1,1)$ and that
there exists a sequence $\{ f_n \} \subset \mathcal{K}$ 
such that $\|f_n\|=1$ and $\lim_{n \to \infty}\|(T-\cos \xi)f_n\| = 0$.
We observe that $\psi_n:=(1-e^{i\xi}S) d_A^* f_n \in \mathcal{L}_1$
and that
\begin{align*}
\|\psi_n\|^2 
	& = 2\|d_A^* f_n\|^2 -2{\rm Re}(e^{i\xi}\langle d_A^* f_n, Sd_A^*f_n \rangle) \\
	& = 2 - 2 \cos \xi  \langle f_n, Tf_n \rangle \\
	& = 2(1-\cos^2 \xi) + o(1).
\end{align*}
Because $\lim\inf_{n\to\infty} \|\psi_n\|^2 = 2(1-\cos^2 \xi) > 0$,
$\psi_n$ does not converge to zero.
Hence, taking a subsequence if needed, 
we can assume that $\inf_n \|\psi_n\| =: c > 0$.
Then we have
\begin{align*}
U \psi_n & = U(1-e^{i\xi}S) d_A^* f_n
		= d_B^* f_n - e^{i\xi}(2d_B^* T -d_A^*)f_n \\
	& =  (1 - 2e^{i\xi} \cos \xi )d_B^* f_n  +  e^{i\xi}d_A^*f_n + o(1) \\
	& = (- e^{2i\xi} S +  e^{i\xi})d_A^*f_n + o(1) \\
	& = e^{i\xi} \psi_n + o(1).
\end{align*}
Let $\tilde{\psi}_n : = \psi_n/\|\psi_n\|$. 
Then, from an argument similar to the above,
we obtain $e^{i\xi} \in \sigma(U_{\overline{\mathcal{L}_1}})$.
Thus $I_1 \subset \sigma(U_{\overline{\mathcal{L}_1}})$ is proved.

Let $\pm 1 \in I_2$. Then $\pm 1 \in \sigma_{\rm c}(T)$
and hence $\pm 1$ can not be an isolated point of $\sigma(T)$.
Hence there exists a sequence $\{ \cos \xi_n \} \subset \sigma(T) \cap (-1, 1)$
such that $\lim_{n \to \infty} \cos \xi_n = \pm 1$.
Because $\lim_{n \to \infty} e^{i\xi_n} = \pm 1$ and $e^{i\xi_n} \in I_1$,
from the above result, we know
that $e^{i\xi_n} \in \sigma(U_{\overline{\mathcal{L}_1}})$.
Because $\sigma(U_{\overline{\mathcal{L}_1}})$ is a closed set,
we have $\pm 1\in \sigma(U_{\overline{\mathcal{L}_1}})$.
Thus $I_2 \subset \sigma(U_{\overline{\mathcal{L}_1}})$ is proved.
}
\end{proof}

\section{Concluding remark}
In this paper, we clarified that the unitary involution of the shift operator 
and coisometry of the boundary map cause the reduction of the spectral analysis 
of the unitary operator to one of the underlying self-adjoint operator. 
This result implies that the spectral mapping theorem can be applied to 
general infinite graphs. 
As is seen in Introduction, 
if the underlying symmetric 
random walk on an infinite graph has only the point spectrum,
{\it e.g.}, the Sierpi\'nski lattice, 
then the induced Grover walk also has only the point spectrum (without continuous spectrum).  
This concludes that the induced Grover walk
exhibits localization for any initial state.  
In a companion paper~\cite{SS}, we clarify a relationship between the spectrum and stochastic behavior of our abstractive quantum walk.

\noindent \\
\noindent \\
{\bf Acknowledgments.} 
We thank Hiromichi Ohno for useful comments. 
YuH's work was supported in part by Japan Society for the
Promotion of Science Grant-in-Aid for Scientific Research (C) 25400208 and 
(A) 15H02055 and for Challenging Exploratory Research 26610025. 
ES's work was partially supported by the Japan-Korea Basic
Scientific Cooperation Program ``Non-commutative Stochastic Analysis: New Prospects of Quantum White Noise and Quantum Walk" (2015-2016). 
ES and AS also acknowledge financial supports of the Grant-in-Aid for Young Scientists (B) of Japan Society for the Promotion of Science (Grants No.16K16637 and No.26800054, respectively).


\end{document}